\documentclass[reqno,tbtags,intlimits,a4paper,oneside,12pt]{amsart}
\usepackage[cp1251]{inputenc}%
\usepackage[english]{babel}
\usepackage{amssymb,upref,calc,url}
\usepackage[mathscr]{eucal}
\usepackage{graphicx}
\usepackage{amsmath,amscd,amsthm,verbatim}
\usepackage[all]{xy}
\usepackage[in]{fullpage}
\newtheorem{thm}{\hskip\parindent Theorem}[section]
\newtheorem{lem}[thm]{\hskip\parindent Lemma}

\newtheorem{cor}[thm]{\hskip\parindent Corollary}

\theoremstyle{definition}

\DeclareMathOperator{\wt}{wt}

\begin{document}

\title{Differentiation of Genus $4$ Hyperelliptic~Functions}
\author{
\vspace{-3pt}
V.\,M.~Buchstaber, E.\,Yu.~Bunkova}
\address{Steklov Mathematical Institute of Russian Academy of Sciences}
\email{buchstab@mi-ras.ru, bunkova@mi-ras.ru}
\thanks{Supported in part by RAS~program "Nonlinear dynamics: fundamental~problems and applications", RFBR project 17-01-00366 A,
and Young Russian Mathematics award}

\begin{abstract}
\vspace{-8pt}

In this work we give an explicit solution to the problem of differentiation of hyperelliptic functions in genus $4$ case.
It is a genus $4$ analogue of the classical result of F.~G.~Frobenius and L.~Stickelberger \cite{FS} in the case of elliptic functions. An explicit solution in the genus $2$ case was given in \cite{B2}.  An explicit solution in the genus $3$ case was given in \cite{B3}.
\end{abstract}

\maketitle

\vspace{-20pt}

\section*{Introduction} 

For a meromorphic function $f$ in $\mathbb{C}^g$ the vector $\omega \in \mathbb{C}^g$ is a period if $f(z+\omega) = f(z)$ for all~$z \in \mathbb{C}^g$.
If the periods of a meromorphic function $f$ form a lattice $\Gamma$ of~rank~$2g$ in~$\mathbb{C}^g$, then $f$ is called an \emph{Abelian function}.
Therefore an Abelian function is a meromorphic~function on the complex torus $T^g = \mathbb{C}^g/\Gamma$. We denote the coordinates in~$\mathbb{C}^g$ by~$z = (z_1, z_3, \ldots, z_{2g-1})$.

We consider hyperelliptic curves of genus $g$ in the model
\[
\mathcal{V}_\lambda = \{(x, y)\in\mathbb{C}^2 \colon
y^2 = x^{2g+1} + \lambda_4 x^{2 g - 1}  + \lambda_6 x^{2 g - 2} + \ldots + \lambda_{4 g} x + \lambda_{4 g + 2}\}. 
\]
The curve depends on the parameters $\lambda = (\lambda_4, \lambda_6, \ldots, \lambda_{4 g}, \lambda_{4 g + 2}) \in \mathbb{C}^{2 g}$.
Let $\mathcal{B} \subset \mathbb{C}^{2g}$ be the subspace of parameters such that the curve $\mathcal{V}_{\lambda}$ is nonsingular for~$\lambda \in \mathcal{B}$.
Then we~have $\mathcal{B} = \mathbb{C}^{2g} \backslash \Sigma$, where $\Sigma$ is~the discriminant hypersurface of the~universal curve.

A \emph{hyperelliptic function of genus} $g$ (see \cite{B2, BEL-12, BEL18}) is a meromorphic function in $\mathbb{C}^g \times \mathcal{B}$,
such that for each $\lambda \in \mathcal{B}$ it's restriction on $\mathbb{C}^g \times \lambda$
is an Abelian function, where the torus~$T^g = \mathbb{C}^g/\Gamma$ is~the~Jacobian $\mathcal{J}_\lambda$ of the curve~$\mathcal{V}_\lambda$.
We denote by $\mathcal{F}$ the field of hyperelliptic functions of genus $g$. For the properties of this field, see \cite{BEL-12, BEL18}. 

We consider the problem of constructing the Lie algebra of derivations of $\mathcal{F}$, i.e. to~find~$3g$ independent differential operators $\mathcal{L}$ such that $\mathcal{L} \mathcal{F} \subset \mathcal{F}$. The exposition to the problem, as well as a general approach to the solution was developed in \cite{BL0, BL}. An~overview is given in \cite{BEL18}. In~\cite{FS, B2, B3} an explicit solution to this problem has been obtained for $g=1,2,3$. In~the present work we give an explicit answer to this problem in the genus $g=4$ case. It~is based on the results of \cite{BB19}.

Let $\mathcal{U}$ be the total space of the bundle $\pi: \mathcal{U} \to \mathcal{B}$ with fiber the Jacobian~$\mathcal{J}_\lambda$ of the curve $\mathcal{V}_\lambda$ over $\lambda \in \mathcal{B}$.
Thus, we can say that hyperelliptic functions of genus~$g$ are meromorphic functions in~$\mathcal{U}$.
By Dubrovin--Novikov Theorem~\cite{DN}, there is a~birational isomorphism between $\mathcal{U}$ and the complex linear space~$\mathbb{C}^{3g}$.

We use the theory of hyperelliptic Kleinian functions (see \cite{ BEL-12, BEL, BEL-97, Baker}, and~\cite{WW} for elliptic functions).
Take the coordinates
$(z, \lambda) = (z_1, z_3, \ldots, z_{2 g -1},$ $\lambda_4, \lambda_6, \ldots, \lambda_{4 g}, \lambda_{4 g + 2})$
in $\mathbb{C}^g \times \mathcal{B} \subset \mathbb{C}^{3g}$.
Let $\sigma(z, \lambda)$ be the hyperelliptic sigma function (or elliptic function in case of genus $g=1$).
We set $\partial_k = \frac{\partial}{\partial z_k}$.
Following~\cite{B2, B3, BEL18}, we use the notation
\begin{equation} \label{note}
\zeta_{k} = \partial_k \ln \sigma(z, \lambda), \qquad
\wp_{k_1, \ldots, k_n} = - \partial_{k_1} \cdots \partial_{k_n} \ln \sigma(z, \lambda), 
\end{equation}
where $n \geqslant 2$, $k_s \in \{ 1, 3, \ldots, 2 g - 1\}$. 
The functions $\wp_{k_1, \ldots, k_n}$ are hyperelliptic functions. The field $\mathcal{F}$ is the field of fractions of the ring of polynomials $\mathcal{P}$ generated by the functions $\wp_{k_1, \ldots, k_n}$, where $n \geqslant 2$, $k_s \in \{ 1, 3, \ldots, 2 g - 1\}$. We note that the~derivations of $\mathcal{F}$ that we~construct are derivations of $\mathcal{P}$.

\section{Lie algebra of vector fields in $\mathcal{B}$} \label{S1}

Let $g \in \mathbb{N}$. Following \cite[Section 4]{A}, we consider $\mathbb{C}^{2g+1}$ with coordinates $(\xi_1, \ldots, \xi_{2g+1})$. Let $H$ be the hyperplane in $\mathbb{C}^{2g+1}$ given by the equation $\sum_{k=1}^{2g+1} \xi_k = 0$. The permutation group $S_{2g+1}$ of coordinates in $\mathbb{C}^{2g+1}$ corresponds to the action of the group $A_{2g}$ on $H$. We~associate a vector $\xi \in H$ with the polynomial
\begin{equation} \label{pol1}
\prod_k (x - \xi_k) = x^{2g+1} + \lambda_4 x^{2 g - 1}  + \lambda_6 x^{2 g - 2} + \ldots + \lambda_{4 g} x + \lambda_{4 g + 2},
\end{equation}
where  $\lambda = (\lambda_4, \lambda_6, \ldots, \lambda_{4 g}, \lambda_{4 g + 2}) \in \mathbb{C}^{2g}$. The orbit space $H/A_{2g}$ is identified with $\mathbb{C}^{2g}$.
We denote the variety of regular orbits in $\mathbb{C}^{2g}$ by $\mathcal{B}$.
Thus, $\mathcal{B} \subset \mathbb{C}^{2g}$ is the subspace of parameters $\lambda$ such that the polynomial \eqref{pol1} has no multiple roots, and $\mathcal{B} = \mathbb{C}^{2g} \backslash \Sigma$, where~$\Sigma$ is the discriminant hypersurface.

The gradient of any $A_{2g}$-invariant polynomial determines a vector field in~$\mathbb{C}^{2g}$ that is~tangent to the discriminant hypersurface $\Sigma$ of the genus $g$ hyperelliptic curve \cite{A, BPol}. Choosing a~multiplicative basis in the ring of $A_{2g}$-invariant polynomials, we can construct the corresponding $2g$ polynomial vector fields, which are linearly independent at each point in $\mathcal{B}$. These fields do not commute and determine a \emph{nonholonomic frame} in~$\mathcal{B}$.

In \cite[Section 4]{A} an approach to constructing an infinite-dimensional Lie algebra of such fields based on the convolution of invariants operation is described.
In the present work, we consider the fields
\[
 L_{0}, \quad L_{2}, \quad L_{4}, \quad  \ldots, \quad L_{4 g - 2},
\]
corresponding to the multiplicative basis in the ring of $A_{2g}$-invariants, that is composed of elementary symmetric functions. The structure polynomials of the convolution of invariants operation in this basis were obtained by D.\,B.\,Fuchs, see~\cite[Section~4]{A}. Note that the~nonholonomic frame in~$\mathcal{B}$ corresponding to the multiplicative basis in the ring of~$A_{2g}$-invariants composed of Newton polynomials is used in the works of V.\,M.\,Buchstaber and~A.\,V.\,Mikhailov, see
\cite{BMinf}. 

We express explicitly the vector fields $\{L_{2k}\}$ in the coordinates $\lambda$. For convenience, we~assume that $\lambda_s = 0$ for all $s \notin \{0,4,6, \ldots, 4 g, 4 g + 2\}$ and $\lambda_0 = 1$.

For $k, m \in \{ 1, 2, \ldots, 2 g\}$, $k\leqslant m$, we set
\[
 T_{2k, 2m} = \sum_{s=0}^{k-1} 2 (k + m - 2 s) \lambda_{2s} \lambda_{2 (k+m-s)}
 - \frac{2 k (2 g - m + 1)}{2 g + 1} \lambda_{2k} \lambda_{2m},
\] 
and for $k > m $ we set $T_{2k, 2m} = T_{2m, 2k}$.

\begin{lem}\label{lem1}
For $k = 0, 1, 2, \ldots, {2 g - 1}$ the formula holds
\begin{equation} \label{Lk}
 L_{2k} = \sum_{s = 2}^{2 g + 1} T_{2k + 2, 2 s - 2} \frac{\partial}{\partial \lambda_{2s}}.
\end{equation} 
\end{lem}

The expressions for the matrix $T = (T_{2k, 2m})$ in \eqref{Lk} are taken from \cite[\S 4]{B3}. A detailed proof of the Lemma can be found in \cite[Lemma 3.1]{4A}.

The vector field $L_0$ is the Euler vector field. It determines the weights of the vector fields $L_k$. Namely, $\wt \lambda_{2k} = 2 k$, $\wt L_{2k} = 2 k$ and for all $k$ we have
\begin{equation} \label{L0}
 [L_0, L_{2k}] = 2 k L_{2k}.
\end{equation}

The vector fields $\{L_{2k}\}$, where $k = 0, 1, 2, \ldots, {2 g - 1}$, generate a graded polynomial~Lie algebra~\cite{BPol}. We denote it by~$\mathscr{L}_b$.
Denote by $c_{2i,2j}^{2s}(\lambda)$ the structure polynomials of $\mathscr{L}_b$, i.e.
\begin{equation} \label{cijs}
[L_{2i}, L_{2j}] = \sum_{s=0}^{2g-1}c_{2i,2j}^{2s}(\lambda) L_{2s}.
\end{equation}
We have the relations
\begin{align*}
c_{2j,2i}^{2s}(\lambda) &= - c_{2i,2j}^{2s}(\lambda), &
c_{0,2k}^{2s}(\lambda) &= 0 \quad \text{for} \quad  s \ne k, &
c_{0,2k}^{2k}(\lambda) &= 2k.
\end{align*}

\begin{lem} \label{lg4}
In the genus $g=4$ case we have the expressions
\begin{align*} 
\begin{pmatrix}
c_{2,2j}^{2s}(\lambda)
 \end{pmatrix}
&= {2 \over 9}
\begin{pmatrix}
12 \lambda_6 & - 12 \lambda_4 & 0 & 9 & 0 & 0 & 0 & 0 \\
10 \lambda_8 & 0 & - 10 \lambda_4 & 0 & 18 & 0 & 0 & 0 \\
8 \lambda_{10} & 0 & 0 & - 8 \lambda_4 & 0 & 27 & 0 & 0 \\
6 \lambda_{12} & 0 & 0 & 0 & - 6 \lambda_4 & 0 & 36 & 0 \\
4 \lambda_{14} & 0 & 0 & 0 & 0 & - 4 \lambda_4 & 0 & 45 \\
2 \lambda_{16} & 0 & 0 & 0 & 0 & 0 & - 2 \lambda_4 & 0 \\
\end{pmatrix} \\
&\text{for} \quad j=2,3,4,5,6,7;\\
\begin{pmatrix}
c_{4,2j}^{2s}(\lambda)
 \end{pmatrix}
&= {2 \over 9}
\begin{pmatrix}
- 9 \lambda_{10} & 15 \lambda_8 & - 15 \lambda_6 & 9 \lambda_4 & 0 & 9 & 0 & 0 \\
- 18 \lambda_{12} & 12 \lambda_{10} & 0 & - 12 \lambda_6 & 18 \lambda_4 & 0 & 18 & 0 \\
- 27 \lambda_{14} & 9 \lambda_{12} & 0 & 0 & - 9 \lambda_6 & 27 \lambda_4 & 0 & 27 \\
- 36 \lambda_{16} & 6 \lambda_{14} & 0 & 0 & 0 & - 6 \lambda_6 & 36 \lambda_4 & 0 \\
- 45 \lambda_{18} & 3 \lambda_{16} & 0 & 0 & 0 & 0 & - 3 \lambda_6 & 45 \lambda_4 \\
\end{pmatrix}\\
&\text{for} \quad j=3,4,5,6,7;\\
\begin{pmatrix}
c_{6,2j}^{2s}(\lambda)
 \end{pmatrix}
&= {2 \over 9}
\begin{pmatrix}
- 9 \lambda_{14} & - 9 \lambda_{12} & 16 \lambda_{10} & - 16 \lambda_8 &  9 \lambda_6 & 9 \lambda_4 & 0 & 9 \\
- 18 \lambda_{16} & - 18 \lambda_{14} & 12 \lambda_{12} & 0 & - 12 \lambda_8 & 18 \lambda_6 & 18 \lambda_4 & 0 \\
- 27 \lambda_{18} & - 27 \lambda_{16} & 8 \lambda_{14} & 0 & 0 & - 8 \lambda_8 & 27 \lambda_6 & 27 \lambda_4 \\
0 & - 36 \lambda_{18} & 4 \lambda_{16} & 0 & 0 & 0 & - 4 \lambda_8 & 36 \lambda_6 \\
\end{pmatrix}\\
&\text{for} \quad j=4,5,6,7;\\
\begin{pmatrix}
c_{8,2j}^{2s}(\lambda)
 \end{pmatrix}
&= {2 \over 9}
\begin{pmatrix}
- 9 \lambda_{18} & - 9 \lambda_{16} & - 9 \lambda_{14} & 15 \lambda_{12} & - 15 \lambda_{10} & 9 \lambda_8 & 9 \lambda_6 & 9 \lambda_4 \\
0 & - 18 \lambda_{18} & - 18 \lambda_{16} & 10 \lambda_{14} & 0 & - 10 \lambda_{10} & 18 \lambda_8 & 18 \lambda_6 \\
0 & 0 & - 27 \lambda_{18} & 5 \lambda_{16} & 0 & 0 & - 5 \lambda_{10} & 27 \lambda_8 \\
\end{pmatrix}\\
&\text{for} \quad j=5,6,7;\\
\begin{pmatrix}
c_{10,2j}^{2s}(\lambda)
 \end{pmatrix}
&= {2 \over 9}
\begin{pmatrix}
0 & 0 & - 9 \lambda_{18} & - 9 \lambda_{16} & 12 \lambda_{14} & -12 \lambda_{12} & 9 \lambda_{10} & 9 \lambda_8 \\
0 & 0 & 0 & - 18 \lambda_{18} & 6 \lambda_{16} & 0 & - 6 \lambda_{12} & 18 \lambda_{10} \\
\end{pmatrix}\\
&\text{for} \quad j=6,7;\\
\begin{pmatrix}
c_{12,14}^{2s}(\lambda)
 \end{pmatrix}
&= {2 \over 9}
\begin{pmatrix}
0 & 0 & 0 & 0 & - 9 \lambda_{18} & 7 \lambda_{16} & - 7 \lambda_{14} & 9 \lambda_{12} \\
\end{pmatrix}.
\end{align*}
\end{lem}

\textit{The proof} is a straightforward check using the explicit expressions \eqref{Lk}.

\begin{cor}\label{c13}
In the genus $g=4$ case for $k=3,4,5,6,7$ we have the expression
 \[
L_{2k} = {1 \over 2 (k-2)} [L_2, L_{2k-2}] + {2 (k - 9) \over 9 (k-2) } \left( \lambda_{2k} L_0 - \lambda_4 L_{2k-4}\right).
\]
\end{cor}
\textit{The proof} is obtained from \eqref{cijs} for $i=1$ using the explicit expression for $c_{2,2j}^{2s}(\lambda)$ from~Lemma \ref{lg4}.

\vfill

\eject

\section{Lie algebra of derivations} \label{c2}
The following explicit form for the operators follows from the general theory developed in~\cite{BB20}. It~is based on the results of~\cite{BB19}. We will give independent proofs.

In the genus $g= 4$ case we set:
\begin{align*}
\mathcal{L}_1 &= \partial_1, \qquad \mathcal{L}_3 = \partial_3,
\qquad \mathcal{L}_5 = \partial_5,
\qquad \mathcal{L}_7 = \partial_7,\\
\mathcal{L}_0 &= L_0
- z_1 \partial_1 - 3 z_3 \partial_3 - 5 z_5 \partial_5 - 7 z_7 \partial_7;
\\
\mathcal{L}_2 &= L_2 
- \zeta_1 \partial_1 + {4 \over 3} \lambda_4 z_{3} \partial_1 - \left(z_1 - {8 \over 9} \lambda_4 z_{5}\right) \partial_3 - \left(3 z_3 - {4 \over 9} \lambda_4 z_{7}\right) \partial_5 - 5 z_5 \partial_7;
\\
\mathcal{L}_4 &= L_4
- \zeta_3 \partial_1 - \zeta_1 \partial_3 - \\ & \quad
+ 2 \lambda_6 z_3 \partial_1
- \left( \lambda_4 z_3
- {4 \over 3} \lambda_6 z_5 \right) \partial_3 
- \left( z_1 + 3 \lambda_4 z_5 
- {2 \over 3} \lambda_6 z_7\right) \partial_5
- \left( 3 z_3 + 5 \lambda_4 z_7 \right) \partial_7;
\\
\mathcal{L}_6 &= L_6
- \zeta_5 \partial_1 - \zeta_3 \partial_3 - \zeta_1 \partial_5 + \\ & \quad
+ {5 \over 3} \lambda_8 z_3 \partial_1 + {16 \over 9} \lambda_8 z_5 \partial_3 - \left( \lambda_4 z_3 + 2 \lambda_6 z_5 -{8 \over 9} \lambda_8 z_7 \right) \partial_5 - \left(z_1 +3 \lambda_4 z_5 + 4 \lambda_6 z_7\right) \partial_7;
\\
\mathcal{L}_8 &= L_8
- \zeta_7 \partial_1 - \zeta_5 \partial_3 - \zeta_3 \partial_5 - \zeta_1 \partial_7 + \left({4 \over 3} \lambda_{10} z_3 - \lambda_{12} z_5\right) \partial_1 + \\ & \quad + {20 \over 9} \lambda_{10} z_5 \partial_3 - \left(\lambda_8 z_5 - {10 \over 9} \lambda_{10} z_7\right) \partial_5 - \left(\lambda_4 z_3 + 2 \lambda_6 z_5 + 3 \lambda_8 z_7\right) \partial_7;
\\
\mathcal{L}_{10} &= L_{10} 
- \zeta_7 \partial_3 - \zeta_5 \partial_5 - \zeta_3 \partial_7  + \\ & \quad
+ (\lambda_{12} z_3 - 2 \lambda_{14} z_5 - \lambda_{16} z_7) \partial_1 + {5 \over 3} \lambda_{12} z_5 \partial_3 + {4 \over 3} \lambda_{12} z_7 \partial_5 - (\lambda_8 z_5 + 2 \lambda_{10} z_7) \partial_7;
\\
\mathcal{L}_{12} &= L_{12}
- \zeta_7 \partial_5 - \zeta_5 \partial_7 +\\ & \quad  + \left({2 \over 3} \lambda_{14} z_3 - 3 \lambda_{16} z_5 - 2 \lambda_{18} z_7\right) \partial_1 + \left({10 \over 9} \lambda_{14} z_5 - \lambda_{16} z_7 \right) \partial_3 + {14 \over 9} \lambda_{14} z_7 \partial_5 - \lambda_{12} z_7 \partial_7;
\\
\mathcal{L}_{14} &= L_{14}
- \zeta_7 \partial_7 + \left({1 \over 3} \lambda_{16} z_3 - 4 \lambda_{18} z_5\right) \partial_1 + \left({5 \over 9} \lambda_{16} z_5 - 2 \lambda_{18} z_7\right) \partial_3 + {7 \over 9} \lambda_{16} z_7 \partial_5.
\end{align*}

We denote the Lie algebra generated by the vector fields $\mathcal{L}_{0}$, $\mathcal{L}_{1}$, $\mathcal{L}_{2}$, $\mathcal{L}_{3}$, $\mathcal{L}_{4}$, $\mathcal{L}_{5}$, $\mathcal{L}_{6}$, $\mathcal{L}_{7}$, $\mathcal{L}_{8}$, $\mathcal{L}_{10}$, $\mathcal{L}_{12}$, and $\mathcal{L}_{14}$ by $\mathscr{L}$.

\begin{thm} \label{t21}
The Lie algebra $\mathscr{L}$ is the Lie algebra of derivations of the field $\mathcal{F}$, i.e.~$\mathcal{L}_{k} \varphi \in \mathcal{F}$ for $\varphi \in \mathcal{F}$.
\end{thm}

\textit{The proof} of this Theorem will be given in Section~\ref{c4}.

\begin{lem}
For the commutators in the Lie algebra $\mathscr{L}$ we have the relations:
\begin{align*}
[\mathcal{L}_0, \mathcal{L}_{k}] &= k \mathcal{L}_{k}, & &k = 0,1,2,3,4,5,6,7,8,10,12,14;\\
[\mathcal{L}_{k}, \mathcal{L}_{m}] &= 0, & &k,m = 1,3,5,7.\\
\end{align*}
\end{lem}

\begin{proof}
We use the explicit expressions for $\mathcal{L}_k$ and the fact that $\mathcal{L}_0$ is the Euler vector field, thus $\mathcal{L}_0 \zeta_k = k \zeta_k$.
\end{proof}

\vfill
 
\begin{lem} \label{l23}
For the commutators in the Lie algebra $\mathscr{L}$ we have the relations:
\begin{align*}
\begin{pmatrix}
[\mathcal{L}_1, \mathcal{L}_2]\\
[\mathcal{L}_1, \mathcal{L}_4]\\
[\mathcal{L}_1, \mathcal{L}_6]\\
[\mathcal{L}_1, \mathcal{L}_8]\\
[\mathcal{L}_1, \mathcal{L}_{10}]\\ 
[\mathcal{L}_1, \mathcal{L}_{12}]\\
[\mathcal{L}_1, \mathcal{L}_{14}] 
\end{pmatrix}
&=
\begin{pmatrix}
\wp_{1,1} & - 1 & 0 & 0\\
\wp_{1,3} &\wp_{1,1} & - 1 & 0\\
\wp_{1,5} & \wp_{1,3} &\wp_{1,1} & - 1\\
\wp_{1,7} & \wp_{1,5} & \wp_{1,3} &\wp_{1,1}\\
0 & \wp_{1,7} & \wp_{1,5} & \wp_{1,3}\\
0 & 0 & \wp_{1,7} & \wp_{1,5}\\
0 & 0 & 0 & \wp_{1,7}\\
\end{pmatrix}
\begin{pmatrix}
\mathcal{L}_1\\
\mathcal{L}_3\\
\mathcal{L}_5\\
\mathcal{L}_7
\end{pmatrix};
\\
\begin{pmatrix}
[\mathcal{L}_3, \mathcal{L}_2]\\
[\mathcal{L}_3, \mathcal{L}_4]\\
[\mathcal{L}_3, \mathcal{L}_6]\\
[\mathcal{L}_3, \mathcal{L}_8]\\
[\mathcal{L}_3, \mathcal{L}_{10}]\\ 
[\mathcal{L}_3, \mathcal{L}_{12}]\\
[\mathcal{L}_3, \mathcal{L}_{14}] 
\end{pmatrix}
&=
\begin{pmatrix}
\wp_{1,3} - \lambda_4 & 0 & - 3 & 0\\
\wp_{3,3} & \wp_{1,3} - \lambda_4 & 0 & - 3\\
\wp_{3,5} & \wp_{3,3} & \wp_{1,3} - \lambda_4 & 0\\
\wp_{3,7} & \wp_{3,5} & \wp_{3,3} & \wp_{1,3} - \lambda_4\\
0 & \wp_{3,7} & \wp_{3,5} &  \wp_{3,3}\\
0 & 0 & \wp_{3,7} & \wp_{3,5}\\
0 & 0 & 0 & \wp_{3,7}
\end{pmatrix}
\begin{pmatrix}
\mathcal{L}_1\\
\mathcal{L}_3\\
\mathcal{L}_5\\
\mathcal{L}_7
\end{pmatrix}
+ {1 \over 3} 
\begin{pmatrix}
 7 \lambda_4\\
 6 \lambda_6\\
 5 \lambda_8\\
 4 \lambda_{10}\\
 3 \lambda_{12}\\
 2 \lambda_{14}\\
 \lambda_{16}\\
\end{pmatrix}
\mathcal{L}_1;
\\
\begin{pmatrix}
[\mathcal{L}_5, \mathcal{L}_2]\\
[\mathcal{L}_5, \mathcal{L}_4]\\
[\mathcal{L}_5, \mathcal{L}_6]\\
[\mathcal{L}_5, \mathcal{L}_8]\\
[\mathcal{L}_5, \mathcal{L}_{10}]\\ 
[\mathcal{L}_5, \mathcal{L}_{12}]\\
[\mathcal{L}_5, \mathcal{L}_{14}] 
\end{pmatrix}
&=
\begin{pmatrix}
\wp_{1,5} & 0 & 0 & - 5\\
\wp_{3,5} &\wp_{1,5} & - 3 \lambda_4 & 0\\
\wp_{5,5} & \wp_{3,5} & \wp_{1,5} - 2 \lambda_6 & - 3 \lambda_4\\
\wp_{5,7} - \lambda_{12} &  \wp_{5,5} & \wp_{3,5} - \lambda_{8} &\wp_{1,5} - 2 \lambda_6\\
- 2 \lambda_{14} & \wp_{5,7} - \lambda_{12} & \wp_{5,5} & \wp_{3,5} - \lambda_{8}\\
- 3 \lambda_{16} & - 2 \lambda_{14} & \wp_{5,7} &  \wp_{5,5}\\
- 4 \lambda_{18} & - 3 \lambda_{16} & 0 & \wp_{5,7} \\
\end{pmatrix}
\begin{pmatrix}
\mathcal{L}_1\\
\mathcal{L}_3\\
\mathcal{L}_5\\
\mathcal{L}_7
\end{pmatrix}
 + {4 \over 9} 
\begin{pmatrix}
 2 \lambda_4\\
 3 \lambda_6\\
 4 \lambda_8\\
 5 \lambda_{10}\\
 6 \lambda_{12}\\
 7 \lambda_{14}\\
 8 \lambda_{16}\\
\end{pmatrix}
\mathcal{L}_3;\\
\begin{pmatrix} 
[\mathcal{L}_7, \mathcal{L}_2]\\
[\mathcal{L}_7, \mathcal{L}_4]\\
[\mathcal{L}_7, \mathcal{L}_6]\\
[\mathcal{L}_7, \mathcal{L}_8]\\
[\mathcal{L}_7, \mathcal{L}_{10}]\\ 
[\mathcal{L}_7, \mathcal{L}_{12}]\\
[\mathcal{L}_7, \mathcal{L}_{14}] 
\end{pmatrix}
&=
\begin{pmatrix}
\wp_{1,7} & 0 & 0 & 0\\
\wp_{3,7} &\wp_{1,7} &  0 & - 5 \lambda_4\\
\wp_{5,7} & \wp_{3,7} & \wp_{1,7} & - 4 \lambda_6 \\
\wp_{7,7} & \wp_{5,7} &  \wp_{3,7} &\wp_{1,7} - 3 \lambda_8\\
- \lambda_{16} & \wp_{7,7}&  \wp_{5,7} & \wp_{3,7} - 2 \lambda_{10}\\
- 2 \lambda_{18} & -\lambda_{16} & \wp_{7,7} &  \wp_{5,7} - \lambda_{12}\\
0 & - 2 \lambda_{18} & - \lambda_{16} & \wp_{7,7} \\
\end{pmatrix}
\begin{pmatrix}
\mathcal{L}_1\\
\mathcal{L}_3\\
\mathcal{L}_5\\
\mathcal{L}_7
\end{pmatrix}
 + {2 \over 9} 
\begin{pmatrix}
 2 \lambda_4\\
 3 \lambda_6\\
 4 \lambda_8\\
 5 \lambda_{10}\\
 6 \lambda_{12}\\
 7 \lambda_{14}\\
 8 \lambda_{16}\\
\end{pmatrix}
\mathcal{L}_5.
\end{align*}
\end{lem}

\textit{The proof} follows from the explicit expressions for $\mathcal{L}_k$.

\begin{lem} \label{l24}
For the commutators in the Lie algebra $\mathscr{L}$ we have the relations:
\begin{align*}
\begin{pmatrix}
[\mathcal{L}_2, \mathcal{L}_4]\\
[\mathcal{L}_2, \mathcal{L}_6]\\
[\mathcal{L}_2, \mathcal{L}_8]\\
[\mathcal{L}_2, \mathcal{L}_{10}]\\ 
[\mathcal{L}_2, \mathcal{L}_{12}]\\
[\mathcal{L}_2, \mathcal{L}_{14}] 
\end{pmatrix} & =
\begin{pmatrix}
 c_{2,2j}^{2s}(\lambda) 
\end{pmatrix}
\begin{pmatrix}
\mathcal{L}_{2k}
\end{pmatrix} + {1 \over 2}
\begin{pmatrix}
- \wp_{1,1,3} & \wp_{1,1,1} & 0 & 0 \\
- \wp_{1,3,3} - \wp_{1,1,5} & \wp_{1,1,3} & \wp_{1,1,1} & 0 \\
- 2 \wp_{1,3,5} - \wp_{1,1,7} & \wp_{1,1,5} & \wp_{1,1,3} & \wp_{1,1,1} \\
- 2 \wp_{1,3,7} - \wp_{1,5,5} & \wp_{1,1,7} & \wp_{1,1,5} & \wp_{1,1,3} \\
- 2 \wp_{1,5,7} & 0 & \wp_{1,1,7} & \wp_{1,1,5} \\
- \wp_{1,7,7} & 0 & 0 & \wp_{1,1,7} 
\end{pmatrix}
\begin{pmatrix}
\mathcal{L}_1\\
\mathcal{L}_3\\
\mathcal{L}_5\\
\mathcal{L}_7
\end{pmatrix}.
\end{align*}
\end{lem}

\begin{proof}
From \cite{BB20} we obtain the expressions for $\mathcal{L}_{2k} \zeta_s$, $k = 1,2,3,4,5,6,7$, $s = 1,3,5,7$.
Provided this, the proof follows from the explicit expressions for $\mathcal{L}_k$.
\end{proof}

\vfill
\eject

\section{Polynomial Lie algebra in $\mathbb{C}^{12}$} \label{c3}

\vspace{-2pt}

Following \cite[Chapter 5]{B3}, we consider the diagram $
 \xymatrix{
	\mathcal{U} \ar[d]^{\pi} \ar@{-->}[r]^{\varphi} & \mathbb{C}^{12} \ar[d]^{p}\\
	\mathcal{B} \ar@{^{(}->}[r] & \mathbb{C}^{8}\\
	}   \label{d}$.
	
Here $\pi: \mathcal{U} \to \mathcal{B}$ is the bundle described above, $\mathcal{B} \subset \mathbb{C}^{8}$ is the embedding given by the coordinates $\lambda$ for $g=4$. We determine the map $\varphi$ by the set of generators of~$\mathcal{F}$. Denote the coordinates in $\mathbb{C}^{12}$ by $x_{i,j}$, where $i \in \{ 1,2,3 \}$ and~$j \in \{1, 3, 5, 7 \}$, and $x_{i+1} = x_{i,1}$.
Then $\varphi$ is determined by the map (see the notation in \eqref{note})
\[
(z,\lambda) \mapsto \begin{pmatrix}
x_2 & x_{1,3} & x_{1,5} & x_{1,7}\\
x_3 & x_{2,3} & x_{2,5} & x_{2,7}\\
x_4 & x_{3,3} & x_{3,5} & x_{3,7}\\
\end{pmatrix} = \begin{pmatrix}
\wp_{1,1} & \wp_{1,3} & \wp_{1,5} & \wp_{1,7}\\
\wp_{1,1,1} & \wp_{1,1,3} & \wp_{1,1,5} & \wp_{1,1,7}\\
\wp_{1,1,1,1} & \wp_{1,1,1,3} & \wp_{1,1,1,5} & \wp_{1,1,1,7}
\end{pmatrix}.
\]
The polynomial map $p$ is determined by the relations \vspace{-3pt}
\begin{align*} 
\lambda_{4} &= - 3 x_{2}^2 + \frac{1}{2} x_{4} - 2 x_{1, 3}, \\
\lambda_{6} &= 2 x_2^3 + \frac{1}{4} x_{3}^2 - \frac{1}{2} x_{2} x_{4} - 2 x_{2} x_{1, 3} + \frac{1}{2} x_{3, 3} - 2 x_{1, 5}, \\
\lambda_{8} &= (4 x_{2}^2 + x_{1, 3}) x_{1, 3} - 2 x_{2} x_{1, 5} + \frac{1}{2} x_{3, 5} - 2 x_{1, 7} - \frac{1}{2} (x_{4} x_{1, 3} - x_{3} x_{2, 3} + x_{2} x_{3, 3}),\\
\lambda_{10} &= 2 x_{2} x_{1, 3}^2 + {1 \over 4} x_{2,3}^2 - \frac{1}{2} x_{1, 3} x_{3, 3} + 2 \left( 2 x_{2}^2 + x_{1, 3} \right) x_{1, 5} - 2 x_{2} x_{1, 7} + \frac{1}{2} x_{3, 7} - \\ & \qquad - 
\frac{1}{2} (x_{4} x_{1, 5} - x_{3} x_{2, 5} + x_{2} x_{3, 5} ),\\
\lambda_{12} &= 4 x_2 x_{1, 3} x_{1, 5} + x_{1, 5}^2 + \left(4 x_{2}^2 + 2 x_{1, 3}\right) x_{1, 7} - \\
& \qquad - \frac{1}{2} \left( x_{3, 3} x_{1, 5} - x_{2,3} x_{2,5}
+ x_{1, 3}x_{3, 5} + x_{4} x_{1, 7} - x_{3} x_{2, 7} + x_{2} x_{3, 7}\right), \\
\lambda_{14} &= 2 x_2 x_{1, 5}^2 + {1 \over 4} x_{2,5}^2 - \frac{1}{2} x_{1, 5} x_{3, 5} + \left( 4 x_2 x_{1,3} + 2 x_{1,5}\right) x_{1,7} - \\
& \qquad - \frac{1}{2} \left( x_{3,3} x_{1,7} - x_{2,3} x_{2,7} + x_{1,3} x_{3,7} \right), \\
\lambda_{16} &= 4 x_2 x_{1, 5} x_{1, 7} + x_{1,7}^2 - \frac{1}{2} \left(x_{3, 5} x_{1,7} - x_{2,5} x_{2,7} +  x_{1,5} x_{3, 7}\right), \\
\lambda_{18} &= 2 x_2 x_{1, 7}^2 + {1 \over 4} x_{2,7}^2 - \frac{1}{2} x_{1,7} x_{3, 7}. 
\end{align*} 
We also have \cite[Proof of Theorem 5.3]{B3} the formulas for $w_{i,j} = \wp_{i,j}$, where $i, j \in {3,5,7}$: \vspace{-3pt}
\begin{align*} 
w_{3, 3} &= 3 x_{2} x_{1, 3} - \frac{1}{2} x_{3, 3} + 3 x_{1, 5}, \\
w_{3, 5} &= 3 x_{2} x_{1, 5} - \frac{1}{2} x_{3, 5} + 3 x_{1, 7}, \\
w_{3, 7} &= 3 x_{2} x_{1, 7} - \frac{1}{2} x_{3, 7},
\\
w_{5, 5} &= \frac{1}{2} (x_{4} x_{1, 5} - x_{3} x_{2, 5} + x_{2} x_{3, 5}) - \left(4 x_{2}^2 + x_{1, 3}\right) x_{1, 5} + 5 x_{2} x_{1, 7} - x_{3, 7},\\
w_{5, 7} &= \frac{1}{2} (x_{4} x_{1,7} - x_{3} x_{2,7} + x_{2} x_{3,7}) - \left(4 x_{2}^2 + x_{1, 3}\right) x_{1, 7},\\
w_{7,7} &= \frac{1}{2} \left( x_{3,3} x_{1,7} - x_{2,3} x_{2,7} + x_{1,3} x_{3,7} \right) - \left( 4 x_2 x_{1,3} + x_{1,5}\right) x_{1,7}.
\end{align*} 

Next we introduce explicitly a set of polynomial vector fields in $\mathbb{C}^{12}$.

We have the polynomial vector fields (see \cite[Lemma 6.2]{B3}):
\begin{align*}
\mathcal{D}_0 &= 2 x_{2} {\partial \over \partial x_{2}} + 3 x_{3} {\partial \over \partial x_{3}} +
4 x_{4} {\partial \over \partial x_{4}} +
4 x_{1,3} {\partial \over \partial x_{1,3}} + 5 x_{2,3} {\partial \over \partial x_{2,3}} +
6 x_{3,3} {\partial \over \partial x_{3,3}} +
\\ & \qquad  + 6 x_{1,5} {\partial \over \partial x_{1,5}} + 7 x_{2,5} {\partial \over \partial x_{2,5}} +
8 x_{3,5} {\partial \over \partial x_{3,5}} +
8 x_{1,7} {\partial \over \partial x_{1,7}} + 9 x_{2,7} {\partial \over \partial x_{2,7}} +
10 x_{3,7} {\partial \over \partial x_{3,7}};
\\
\mathcal{D}_1 &= x_{3} {\partial \over \partial x_{2}} + x_{4} {\partial \over \partial x_{3}} + 4 (3 x_{2} x_{3} + x_{2,3}) {\partial \over \partial x_{4}}
+ \\
& \qquad + \sum_{s = 3,5,7} x_{2,s} {\partial \over \partial x_{1,s}} + x_{3,s} {\partial \over \partial x_{2,s}} + 4 (x_{3} x_{1,s} + 2 x_{2} x_{2,s} + x_{2,s+2}) {\partial \over \partial x_{3,s}};
\end{align*}
where $x_{2,9} = 0$.
We denote by $x_{3,k}'$ the expression $\mathcal{D}_1(x_{3,k})$, by $w_{k,s}'$ the~expression $\mathcal{D}_1(w_{k,s})$, by $w_{k,s}''$ the~expression $\mathcal{D}_1(\mathcal{D}_1(w_{k,s}))$, and by $w_{k,s}'''$ the expression $\mathcal{D}_1(\mathcal{D}_1(\mathcal{D}_1(w_{k,s})))$.
Then we~have the~polynomial vector fields (see \cite[Lemma 6.3]{B3}) for $k = 3,5,7$:
\begin{align*}
\mathcal{D}_{k} &= x_{2,k} {\partial \over \partial x_{2}} + x_{3,k} {\partial \over \partial x_{3}} 
+ x_{3,k}' {\partial \over \partial x_{4}} +
\sum_{s = 3,5,7} w_{k,s}' {\partial \over \partial x_{1,s}} + w_{k,s}'' {\partial \over \partial x_{2,s}}
+ w_{k,s}''' {\partial \over \partial x_{3,s}}.
\end{align*}

Consider the vector fields $\mathcal{D}_2$ and $\mathcal{D}_4$ determined by the conditions 
\begin{align*}
\mathcal{D}_2 x_2 &= - x_2^2 + {1 \over 2} x_4 + 2 x_{1,3} + {7 \over 9} \lambda_4;\\
\mathcal{D}_2 x_{1,3} &= - x_2 x_{1,3} + {1 \over 2} x_{3,3} + 3 x_{1,5} - {4 \over 3} \lambda_4 x_2 + w_{3,3};\\
\mathcal{D}_2 x_{1,5} &= - x_2 x_{1,5} + {1 \over 2} x_{3,5} + 5 x_{1,7} - {8 \over 9} \lambda_4 x_{1,3} + w_{3,5};\\
\mathcal{D}_2 x_{1,7} &= - x_2 x_{1,7} + {1 \over 2} x_{3,7} - {4 \over 9} \lambda_4 x_{1,5} + w_{3,7};\\
\mathcal{D}_4 x_2 &= - 2 x_2 x_{1,3}  + x_{3,3} + 2 x_{1,5} + {2 \over 3} \lambda_6;\\
\mathcal{D}_4 x_{1,3} &= - x_{1,3}^2 + 3 x_{1,7} + \lambda_4 x_{1,3} - 2 \lambda_6 x_2 - \lambda_8 - x_2 w_{3,3} + w_{3,3}'' + w_{3,5};\\
\mathcal{D}_4 x_{1,5} &= - x_{1,3} x_{1,5} + 3 \lambda_4 x_{1,5} - {4 \over 3} \lambda_6 x_{1,3} - x_2 w_{3,5} + w_{3,5}'' + w_{5,5};\\
\mathcal{D}_4 x_{1,7} &= - x_{1,3} x_{1,7} + 5 \lambda_4 x_{1,7} - {2 \over 3} \lambda_6 x_{1,5} - x_2 w_{3,7} + w_{3,7}'' + w_{5,7};
\end{align*}
and by the relations $
\begin{pmatrix} 
[\mathcal{D}_1, \mathcal{D}_2]\\
[\mathcal{D}_1, \mathcal{D}_4]\\
\end{pmatrix}
=
\begin{pmatrix}
x_2 & - 1 & 0\\
x_{1,3} & x_2 & - 1\\
\end{pmatrix}
\begin{pmatrix}
\mathcal{D}_1\\
\mathcal{D}_3\\
\mathcal{D}_5
\end{pmatrix}.
$

We set (see Corollary \ref{c13} and Lemma \ref{l24}):
\begin{align*}
\mathcal{D}_6 &= {1 \over 4} \left(2 [\mathcal{D}_2, \mathcal{D}_4] + x_{2, 3} \mathcal{D}_1 - x_{3} \mathcal{D}_3 \right) - {4 \over 3 } \left( \lambda_{6} \mathcal{D}_0 - \lambda_4 \mathcal{D}_2\right);\\
\mathcal{D}_8 &= {1 \over 8} \left(2 [\mathcal{D}_2, \mathcal{D}_6] + (w_{3, 3}' + x_{2,5}) \mathcal{D}_1 - x_{2,3} \mathcal{D}_3 - x_{3} \mathcal{D}_5)\right) - {5 \over 9 } \left( \lambda_{8} \mathcal{D}_0 - \lambda_4 \mathcal{D}_4\right);
\\
\mathcal{D}_{10} &= {1 \over 12} \left(2 [\mathcal{D}_2, \mathcal{D}_8] + ( 2 w_{3,5}'+ x_{2,7}) \mathcal{D}_1 - x_{2,5} \mathcal{D}_3 - x_{2,3} \mathcal{D}_5 - x_{2,1} \mathcal{D}_7\right) - {8 \over 27 } \left( \lambda_{10} \mathcal{D}_0 - \lambda_4 \mathcal{D}_6\right);
\\
\mathcal{D}_{12} &= {1 \over 16} \left(2 [\mathcal{D}_2, \mathcal{D}_{10}] + (2 w_{3,7}' + w_{5,5}') \mathcal{D}_1 - x_{2,7} \mathcal{D}_3 - x_{2,5} \mathcal{D}_5 - x_{2,3} \mathcal{D}_7\right) - { 1 \over 6 } \left( \lambda_{12} \mathcal{D}_0 - \lambda_4 \mathcal{D}_8\right);\\
\mathcal{D}_{14} &= {1 \over 20} \left(2 [\mathcal{D}_2, \mathcal{D}_{12}] + 2 w_{5, 7}' \mathcal{D}_1 - x_{2,7} \mathcal{D}_5 - x_{2,5} \mathcal{D}_7 \right) - {4 \over 45 } \left( \lambda_{14} \mathcal{D}_0 - \lambda_4 \mathcal{D}_{10}\right).
\end{align*}

Denote the ring of polynomials in $\lambda \in \mathbb{C}^{8}$ by $P$.
Let us consider the polynomial map~$p\colon \mathbb{C}^{12} \to \mathbb{C}^{8}$.
A vector field $\mathcal{D}$ in $\mathbb{C}^{12}$ will be called projectable for $p$ if there exists a vector field $L$ in $\mathbb{C}^{8}$ such that
\[
 \mathcal{D}(p^* f) = p^* L(f) \quad \text{for any} \quad f \in P.
\]
The vector field $L$ will be called the pushforward of $\mathcal{D}$.

\begin{lem}
The vector fields $\mathcal{D}_s$, $s=1,3,5,7$ are projectable for $p$ with trivial pushforwards, i.e.
\[
\mathcal{D}_s(f) = 0, \qquad s=1,3,5,7, \quad \text{for any} \quad f \in P.
\]
The vector fields $\mathcal{D}_{2k}$, $k=0,1,2,3,4,5,6,7$, are projectable for $p$ and~their pushforwards are $L_{2k}$, i.e.
\[
 \mathcal{D}_{2k}(p^* f) = p^* L_{2k}(f), \qquad k = 0,1,2,3,4,5,6,7, \quad \text{for any} \quad f \in P.
\]
\end{lem}
\begin{proof}
The formulas above define all the expressions involved explicitly, so the check is a direct calculation.
\end{proof}

We denote the polynomial Lie algebra generated by the vector fields $\mathcal{D}_{0}$, $\mathcal{D}_{1}$, $\mathcal{D}_{2}$, $\mathcal{D}_{3}$, $\mathcal{D}_{4}$, $\mathcal{D}_{5}$, $\mathcal{D}_{6}$, $\mathcal{D}_{7}$, $\mathcal{D}_{8}$, $\mathcal{D}_{10}$, $\mathcal{D}_{12}$, and $\mathcal{D}_{14}$ by $\mathscr{D}$.

\begin{lem} \label{d4}
For the commutators in the Lie algebra $\mathscr{D}$ we have the relations:
\begin{align*}
\begin{pmatrix}
[\mathcal{D}_1, \mathcal{D}_2]\\
[\mathcal{D}_1, \mathcal{D}_4]\\
[\mathcal{D}_1, \mathcal{D}_6]\\
[\mathcal{D}_1, \mathcal{D}_8]\\
[\mathcal{D}_1, \mathcal{D}_{10}]\\ 
[\mathcal{D}_1, \mathcal{D}_{12}]\\
[\mathcal{D}_1, \mathcal{D}_{14}] 
\end{pmatrix}
&=
\begin{pmatrix}
x_2 & - 1 & 0 & 0\\
x_{1,3} & x_2 & - 1 & 0\\
x_{1,5} & x_{1,3} & x_2 & - 1\\
x_{1,7} & x_{1,5} & x_{1,3} & x_2\\
0 & x_{1,7} & x_{1,5} & x_{1,3}\\
0 & 0 & x_{1,7} & x_{1,5}\\
0 & 0 & 0 & x_{1,7}\\
\end{pmatrix}
\begin{pmatrix}
\mathcal{D}_1\\
\mathcal{D}_3\\
\mathcal{D}_5\\
\mathcal{D}_7
\end{pmatrix};
\\
\begin{pmatrix}
[\mathcal{D}_3, \mathcal{D}_2]\\
[\mathcal{D}_3, \mathcal{D}_4]\\
[\mathcal{D}_3, \mathcal{D}_6]\\
[\mathcal{D}_3, \mathcal{D}_8]\\
[\mathcal{D}_3, \mathcal{D}_{10}]\\ 
[\mathcal{D}_3, \mathcal{D}_{12}]\\
[\mathcal{D}_3, \mathcal{D}_{14}] 
\end{pmatrix}
&=
\begin{pmatrix}
x_{1,3} - \lambda_4 & 0 & - 3 & 0\\
w_{3,3} & x_{1,3} - \lambda_4 & 0 & - 3\\
w_{3,5} & w_{3,3} & x_{1,3} - \lambda_4 & 0\\
w_{3,7} & w_{3,5} & w_{3,3} & x_{1,3} - \lambda_4\\
0 & w_{3,7} & w_{3,5} & w_{3,3}\\
0 & 0 & w_{3,7} & w_{3,5}\\
0 & 0 & 0 & w_{3,7}
\end{pmatrix}
\begin{pmatrix}
\mathcal{D}_1\\
\mathcal{D}_3\\
\mathcal{D}_5\\
\mathcal{D}_7
\end{pmatrix}
+ {1 \over 3} 
\begin{pmatrix}
 7 \lambda_4\\
 6 \lambda_6\\
 5 \lambda_8\\
 4 \lambda_{10}\\
 3 \lambda_{12}\\
 2 \lambda_{14}\\
 \lambda_{16}\\
\end{pmatrix}
\mathcal{D}_1;
\\
\begin{pmatrix}
[\mathcal{D}_5, \mathcal{D}_2]\\
[\mathcal{D}_5, \mathcal{D}_4]\\
[\mathcal{D}_5, \mathcal{D}_6]\\
[\mathcal{D}_5, \mathcal{D}_8]\\
[\mathcal{D}_5, \mathcal{D}_{10}]\\ 
[\mathcal{D}_5, \mathcal{D}_{12}]\\
[\mathcal{D}_5, \mathcal{D}_{14}] 
\end{pmatrix}
&=
\begin{pmatrix}
x_{1,5} & 0 & 0 & - 5\\
w_{3,5} & x_{1,5} & - 3 \lambda_4 & 0\\
w_{5,5} & w_{3,5} & x_{1,5} - 2 \lambda_6 & - 3 \lambda_4\\
w_{5,7} - \lambda_{12} &  w_{5,5} & w_{3,5} - \lambda_{8} & x_{1,5} - 2 \lambda_6\\
- 2 \lambda_{14} & w_{5,7} - \lambda_{12} & w_{5,5} & w_{3,5} - \lambda_{8}\\
- 3 \lambda_{16} & - 2 \lambda_{14} & w_{5,7} & w_{5,5}\\
- 4 \lambda_{18} & - 3 \lambda_{16} & 0 & w_{5,7} \\
\end{pmatrix}
\begin{pmatrix}
\mathcal{D}_1\\
\mathcal{D}_3\\
\mathcal{D}_5\\
\mathcal{D}_7
\end{pmatrix}
 + {4 \over 9} 
\begin{pmatrix}
 2 \lambda_4\\
 3 \lambda_6\\
 4 \lambda_8\\
 5 \lambda_{10}\\
 6 \lambda_{12}\\
 7 \lambda_{14}\\
 8 \lambda_{16}\\
\end{pmatrix}
\mathcal{D}_3;
\\
\begin{pmatrix} 
[\mathcal{D}_7, \mathcal{D}_2]\\
[\mathcal{D}_7, \mathcal{D}_4]\\
[\mathcal{D}_7, \mathcal{D}_6]\\
[\mathcal{D}_7, \mathcal{D}_8]\\
[\mathcal{D}_7, \mathcal{D}_{10}]\\ 
[\mathcal{D}_7, \mathcal{D}_{12}]\\
[\mathcal{D}_7, \mathcal{D}_{14}] 
\end{pmatrix}
&=
\begin{pmatrix}
x_{1,7} & 0 & 0 & 0\\
w_{3,7} & x_{1,7} &  0 & - 5 \lambda_4\\
w_{5,7} & w_{3,7} & x_{1,7} & - 4 \lambda_6 \\
w_{7,7} & w_{5,7} & w_{3,7} & x_{1,7} - 3 \lambda_8\\
- \lambda_{16} & w_{7,7} & w_{5,7} & w_{3,7} - 2 \lambda_{10}\\
- 2 \lambda_{18} & -\lambda_{16} & w_{7,7} &  w_{5,7} - \lambda_{12}\\
0 & - 2 \lambda_{18} & - \lambda_{16} & w_{7,7} \\
\end{pmatrix}
\begin{pmatrix}
\mathcal{D}_1\\
\mathcal{D}_3\\
\mathcal{D}_5\\
\mathcal{D}_7
\end{pmatrix}
 + {2 \over 9} 
\begin{pmatrix}
 2 \lambda_4\\
 3 \lambda_6\\
 4 \lambda_8\\
 5 \lambda_{10}\\
 6 \lambda_{12}\\
 7 \lambda_{14}\\
 8 \lambda_{16}\\
\end{pmatrix}
\mathcal{D}_5.
\end{align*}
\end{lem}
\textit{The proof} is a direct calculation.

\vfill
\eject

\section{Proof of Theorem \ref{t21}} \label{c4}

The polynomial Lie algebra $\mathscr{D}$ is a solution to \cite[Problem 6.1]{B3} for $g=4$, namely the polynomial vector fields $\mathcal{D}_s$, $s= 0,1,2,3,4,5,6,7,8,10,12,14$, are projectable for $p$ and independent at any point in $p^{-1}(\mathcal{B})$. In \cite[Chapter 6]{B3}, the connection of this problem with the problem of constructing the Lie algebra of derivations of $\mathcal{F}$ is described. Namely, a solution is given by the differential operators $\mathcal{L}_s$ such that 
\[
\mathcal{L}_s \varphi^* x_{i,j} = \varphi^* \mathcal{D}_s x_{i,j}
\]
for the coordinate functions $x_{i,j}$ in $\mathbb{C}^{12}$, see Section \ref{c3}.

By the construction (see \cite{B3}), we have $\mathcal{L}_s = \partial_s$ for $s = 1,3,5,7$ and $\mathcal{L}_0$ is the Euler vector field. This coincides with the operators presented in Section \ref{c2}. 
The vector fields~$\mathcal{L}_{2k}$, where $k = 2,3,4,5,6,7$, are determined by the conditions:
\begin{enumerate}
\item $\mathcal{L}_{2k} = L_{2k} + f_{2k,1}(u, \lambda) \partial_{1} + f_{2k,3}(z, \lambda) \partial_{3} + f_{2k,5}(z, \lambda) \partial_{5} + f_{2k,7}(z, \lambda) \partial_{7}$,
\item the fields $\mathcal{L}_{2k}$ satisfy commutation relations obtained by $\varphi^*$ from the commutation relations of Lemma \ref{d4}.
\end{enumerate}
By comparing Lemma \ref{d4} with Lemma \ref{l23} we see that for the operators from Section~\ref{c2} these conditions are satisfied.
The condition (2) determines the coefficients $f_{2k,j}(u, \lambda)$, where $j = 1,3,5,7$, up to constants \mbox{in $z_1$, $z_3$, $z_5$, $z_7$.}
The grading of the variables and the condition~$[\mathcal{L}_0, \mathcal{L}_k] = k \mathcal{L}_k$ determines these constants.
\hfill $\square$

\vfill

\eject


\begin{thebibliography}{99}

\bibitem{FS}
\textsc{F.~G.~Frobenius, L.~Stickelberger}, 
\emph{\"Uber die Differentiation der elliptischen Functionen nach den Perioden und Invarianten},
J. Reine Angew. Math., 92 (1882), 311--337.

\bibitem{B2}
\textsc{V.\,M.\,Buchstaber},
\emph{Polynomial dynamical systems and Korteweg--de Vries equation},
Proc. Steklov Inst. Math., 294 (2016), 176--200.

\bibitem{B3}
\textsc{E.\,Yu.\,Bunkova},
\emph{Differentiation of genus 3 hyperelliptic functions},
European Journal of Mathematics, 4:1 (2018), 93--112, arXiv:1703.03947.

\bibitem{BEL-12}
\textsc{V.\,M.\,Buchstaber, V.\,Z.\,Enolskii, D.\,V.\,Leikin},
\emph{Multi-Dimensional Sigma-Functions},
arXiv: 1208.0990, 2012, 267 pp.

\bibitem{BEL18}
\textsc{V.\,M.\,Buchstaber, V.\,Z.\,Enolski, D.\,V.\,Leykin},
\emph{Sigma-functions: old and new results}, Algebraic Geometry, v. 2, LMS Lecture Note Series, Cambridge Univ. Press, 2019, arXiv: 1810.11079.


\bibitem{BL0}
\textsc{V.\,M.\,Buchstaber, D.\,V.\,Leikin},
\emph{Differentiation of Abelian functions with respect to parameters},
Russian Math. Surveys, 62:4 (2007), 787--789.

\bibitem{BL}
\textsc{V.\,M.\,Buchstaber, D.\,V.\,Leikin},
\emph{Solution of the Problem of Differentiation of Abelian Functions over Parameters for Families of $(n,s)$-Curves},
Funct. Anal. Appl., 42:4,
2008, 268--278.

\bibitem{BB19}
\textsc{V.\,M.\,Buchstaber, E.\,Yu.\,Bunkova}, \emph{Lie Algebras of Heat Operators in Nonholonomic Frame}, 2019, arXiv:1911.08266.

\bibitem{DN}
\textsc{B.\,A.\,Dubrovin, S.\,P.\,Novikov},
\emph{A periodic problem for the Korteweg-de Vries and Sturm-Liouville equations. Their connection with algebraic geometry. (Russian)}
Dokl. Akad. Nauk SSSR, 219:3, 1974, 531--534.

\bibitem{BEL}
\textsc{V.~M.~Buchstaber, V.~Z.~Enolskii, D.~V.~Leikin},
\emph{Kleinian functions, hyperelliptic Jacobians and applications},
Reviews in Mathematics and Math. Physics, 10:2, Gordon and Breach, London, 1997, 3--120.

\bibitem{BEL-97}
\textsc{V.\,M.\,Buchstaber, V.\,Z.\,Enolskii, D.\,V.\,Leikin},
\emph{Hyperelliptic Kleinian functions and applications},
``Solitons, Geometry and Topology: On the Crossroad'', Adv. Math. Sci., AMS Transl., 179:2, Providence, RI, 1997, 1--34.

\bibitem{Baker}
\textsc{H.\,F.\,Baker},
\emph{On the hyperelliptic sigma functions},
Amer. Journ. Math. 20, 1898, 301--384.

\bibitem{WW}
\textsc{E.\,T.\,Whittaker, G.\,N.\,Watson}, 
\emph{A Course of Modern Analysis},
Reprint of 4th (1927) ed., Vol 2. \emph{Transcendental functions},
Cambridge Univ. Press, Cambridge, 1996.

\bibitem{A}
\textsc{V.\,I.\,Arnold},
\emph{Singularities of Caustics and Wave Fronts},
Mathematics and its Applications, vol. 62, Kluwer Academic Publisher Group, Dordrecht, 1990.

\bibitem{BPol}
\textsc{V.\,M.\,Buchstaber, D.\,V.\,Leikin},
\emph{Polynomial Lie Algebras},
Funct. Anal. Appl., 36:4 (2002), 267--280.

\bibitem{BMinf}
\textsc{V.\,M.\,Buchstaber, A.\,V.\,Mikhailov},
\emph{Infinite-Dimensional Lie Algebras Determined by the Space of Symmetric Squares of Hyperelliptic Curves},
Functional Analysis and Its Applications, Vol. 51, No.~1, 2017.


\bibitem{4A}
\textsc{J.\,C.\,Eilbeck, J.\,Gibbons, Y.\,Onishi, S.\,Yasuda},
\emph{Theory of Heat Equations for Sigma Functions},
arXiv:1711.08395, (2018).

\bibitem{BB20}
\textsc{V.\,M.\,Buchstaber, E.\,Yu.\,Bunkova}, \emph{Explicit Heat Equations in Nonholomic Frame and Applications}, 2020 (to appear).

\end{thebibliography}
\end{document}